\documentclass[a4paper,11pt,twoside]{article}
\usepackage[INRIA]{lipinria}
\usepackage[latin1]{inputenc}
\usepackage{enumerate}
\usepackage{amsfonts,amsmath,amssymb,amsthm,amstext,latexsym,paralist}
\usepackage{url}
\usepackage{xspace}
\usepackage{subfigure}

\usepackage{a4wide}
\usepackage{algorithm2e}
\usepackage{algorithmic}

\usepackage{figlatex}

\newtheorem{theorem}{Theorem}

\AtBeginDocument{}

\newcommand{\TT}{\mathcal{T}} 
\newcommand{\NN}{\mathcal{N}} 
\newcommand{\CC}{\mathcal{C}} 
\newcommand{\LL}{\mathcal{L}} 
\newcommand{\qos}{\textsf{q}}

\newcommand{\BW}{\textsf{BW}}

\newcommand{\child}{\textsf{children}}

\newcommand{\pat}[2]{\textsf{path}[{#1} \to {#2}]}
\newcommand{\REP}{\textsc{Replica Placement}\xspace}

\newcommand{\CLOSEST}{\textit{Closest}\xspace}

\newcommand{\XX}{\textit{ORP}\xspace}

\usepackage{RR}
\begin{document}
\RRInumber{2007-10}
\RRNo{6233}
\RRItitle{Optimal Replica Placement in Tree Networks with QoS and Bandwidth Constraints and the Closest Allocation Policy}
\RRItitre{Placement optimal de r\'epliques dans les r\'eseaux en arbre avec des contraintes de qualit\'e de service et de bande passante et avec la politique d'allocation \CLOSEST}

\RRIdate{June 2007}

\RRIauthor{Veronika Rehn-Sonigo}

\RRIthead{Optimal Replica Placement in Tree Networks}
\RRIahead{V. Rehn-Sonigo}

\RRIkeywords{Replica placement, tree networks, \CLOSEST policy, quality of service, bandwidth constraints.}

\RRImotscles{Placement de r\'epliques, r\'eseaux en arbres, politique \CLOSEST, qualit\'e de service, limitation de bande passante.}

\RRIabstract{This paper deals with the replica placement problem on fully homogeneous tree networks known as the \REP optimization problem. The client requests are known beforehand, while the number and location of the servers are to be determined. We investigate the latter problem using the \CLOSEST access policy when adding QoS and bandwidth constraints. We propose an optimal algorithm in two passes using dynamic programming.}

\RRIresume{Dans ce papier, nous traitons le probl\`eme du placement de r\'epliques dans des r\'eseaux en arbres compl\`etement homog\`enes, connu sous le nom du probl\`eme d'optimisation de \REP. Les requ\^etes des clients sont connues {\it a priori}, mais le nombre et les emplacements des serveurs restent \`a d\'eterminer. Nous \'etudions ce dernier probl\`eme en utilisant la politique d'acc\`es \CLOSEST, en ajoutant de la qualit\'e de service et des limitations de bande passante. Nous proposons un algorithme optimal en deux passes qui utilise la programmation dynamique.}

\RRItheme{\THNum}
\RRIprojet{GRAAL}
\RRNo{6233}

\RRImaketitle

\section{Introduction}
\label{sec:intro}

This paper deals with the problem of replica placement in tree networks with
Quality of Service (QoS) guarantees and bandwidth constraints.
Informally, there are clients issuing several requests per time-unit, to be satisfied by
servers with a given QoS and respecting the bandwidth limits of the interconnection links. The clients
are known (both their position in the tree and their number of
requests), while the number and location of the servers are to be
determined. A client is a leaf node of the tree, and its requests can
be served by one or several internal nodes.
Initially, there are no replicas; when a node is equipped with a
replica, it can process a number
of requests, up to its capacity limit (number of requests served by time-unit).
Nodes equipped with a
replica, also called servers, can only serve clients located in
their subtree (so that the root, if equipped with a replica, can serve
any client); this restriction is usually adopted to enforce the
hierarchical nature of the target application platforms, where a node
has knowledge only of its parent and children in the tree. Every
client has some QoS constraints: its requests must be served within a
limited time, and thus the servers handling these requests must not be
too far from the client.

The rule of the game is to assign replicas to internal nodes so that some optimization
function is minimized and QoS as well as bandwidth constraints are respected. Typically,
this optimization function is the total
utilization cost of the servers. We restrict the problem to the most popular access policy called \CLOSEST, where
each client is allowed to be served only by the closest replica in the path from itself up to the root.

In this paper we study this optimization problem, called \REP, and we
restrict the QoS in terms of number of hops.
This means for instance that the requests of a client who has a QoS range of
$5$ must be treated by one of the first five internal nodes on the path from
the client up to the tree root.

We point out that the distribution tree (clients and nodes) is fixed
in our approach.
This key assumption is quite natural for a broad spectrum of
applications, such as electronic, ISP, or VOD
service delivery. The root server has the original copy of the database
but cannot serve all clients directly,
so a distribution tree is deployed to provide a hierarchical and
distributed access to replicas of the
original data. On the contrary, in other, more decentralized, applications
(e.g. allocating Web mirrors in distributed
networks), a two-step approach is used: first determine a ``good'' distribution
tree in an arbitrary interconnection graph, and then determine a ``good'' placement of replicas among the tree
nodes. Both steps are interdependent, and the problem is much more complex, due to the combinatorial
solution space (the number of candidate distribution trees may well be exponential).

Many authors deal with the \REP optimization problem. Most of the
papers 
neither deal with QoS nor with bandwidth constraints. Instead they consider average system performance as
total communication cost or total accessing cost. Please refer
to~\cite{rr2006-30} for a detailed description of related work with no
QoS constraints.

Cidon et al.~\cite{Cidon2002} studied an instance of \REP with
multiple objects, where all requests of a client are served by the
closest replica (\CLOSEST policy). In this work, the objective
function 
integrates a communication cost, which
can be seen as a
substitute for QoS. Thus, they minimize the average
communication cost for all the clients rather than ensuring a given
QoS for each client. They target fully homogeneous platforms since
there are no server capacity constraints in their approach.
A similar instance of the problem has been studied by Liu et
al~\cite{PangfengLiu06}, adding a QoS in terms of a range limit,
and whose objective is to minimize the number of
replicas. In this latter approach, the servers are homogeneous, and
their capacity is bounded. Both~\cite{Cidon2002,PangfengLiu06} use
a dynamic programming algorithm to find the optimal solution.

Some of the first authors to introduce actual QoS constraints in the
problem were Tang and Xu~\cite{Tang2005}. In
their approach, the QoS corresponds to the latency requirements of
each client. Different access policies are considered. First, a
replica-aware policy in a general graph with heterogeneous nodes
 is proven to be NP-complete.
When the clients do not know where the replicas are (replica-blind
policy), the graph is simplified to a tree (fixed routing scheme)
with the \CLOSEST policy, and in this case again it is possible to
find an optimal dynamic programming algorithm.

Bandwidth limitations are taken into account when Karlsson et al.~\cite{Karlsson02,Karlsson04} compare different
objective functions and several heuristics to solve NP-complete problem instances. They do not take QoS
constraints into account, but instead integrate a communication cost in the objective function as
was done in~\cite{Cidon2002}. Integrating the communication cost into the objective function can be viewed as
a Lagrangian relaxation of QoS constraints.
Please refer to~\cite{rr2006-48} for more related work dealing with
QoS constraints.

In this paper we propose an efficient algorithm called 
\textbf{Optimal Replica Placement} (\XX) to determine optimal 
locations for placing replicas in the \REP problem including QoS 
and bandwidth. Our work provides a major extension
of the algorithm of Liu et al.~\cite{PangfengLiu06}, which was 
already mentioned above.
Liu et al.~\cite{PangfengLiu06} proposed an algorithm 
\textbf{Place-replica} to find an optimal set of replicas on 
homogeneous data grid trees including QoS constraints in terms of 
distance but without bandwidth constraints.
Our approach leads to two important extensions. First of all, we 
separate the set of clients from the set of servers, while Liu et 
al suppose clients to be servers with a double functionality. Our 
model can simulate the latter model
while the converse is not true.  Indeed, we can model 
client-server nodes by inserting a fictive node before the client 
which can take the role of a server. The approach of Liu et al. in 
contrast does not offer the possibility to model clients without 
server functionality.

Our second major contribution is the introduction of bandwidth 
constraints. This is an important modification of the requirements 
as QoS and bandwidth are of a completely different nature. QoS is 
a constraint that belongs to a node locally, hence each client has 
to cope with its own limitation. Bandwidth constraints in contrast 
have a global influence on the resources as a link may be shared 
by multiple clients and consequently all of them are concerned. 
Therefore it is not obvious whether the problem with these 
completely different constraint types would remain  polynomial or 
would become NP-hard.

The rest of the paper is organized as follows. 
Section~\ref{sec:notations}
introduces our main notations used in \REP problems. 
Section~\ref{sec:algo} is dedicated to the presentation of our 
polynomial algorithm: the proper terminology of the algorithm is 
introduced in Section~\ref{sec:terminology}. The subsections~\ref{sec:phase1} 
and~\ref{sec:phase2} treat the different 
phases and explaining examples can be found in 
Sections~\ref{sec:ex1} and~\ref{sec:ex2}. Complexity is subject of 
Section~\ref{sec:complexity}, whereas optimality is proven in 
Section~\ref{sec:optimality}. Section~\ref{sec:conclusion} 
summarizes our work.


\section{Notations}
\label{sec:notations}
This section familiarizes with our basic notations.
We consider a distribution tree $\TT$ whose nodes are partitioned 
into a set of clients $\CC$
and a set of internal nodes $\NN$ ($\NN \cap \CC = \emptyset$).
The clients are leaf nodes of the tree, while $\NN$ is the set of
internal nodes. Let $r$ be the root of the tree.  The set of tree 
edges (links) is denoted as $\LL$. Each link $l$ owns a bandwidth 
limit $\BW(l)$ that can not be exceeded.

A \emph{client} $v \in \CC$ is making $w_v$ requests per time unit
to a database. Each client has to respect its personal 
{\it Quality of Service} constraints (QoS), where $\qos(v)$ 
indicates the range limit in hops for $v$ upwards to the root 
until a database replica has to be reached. 
%
A \emph{node} $j\in \NN$ may or may not have been provided with a 
replica of the database.
Nodes equipped with a replica ({\em i.e.} servers) can process up 
to
$W$ requests per time unit from clients in their subtree. In other
words, there is a unique path from a client $v$ to the root of the 
tree, and each node in this path
is eligible to process all the requests issued by $v$ when
provided with a replica. We denote by $R \subseteq \NN$ the entire
set of nodes equipped with a replica.




\section{Optimal Replica Placement Algorithm (\XX)}
\label{sec:algo}
In this section we present \XX, an algorithm to solve the \REP 
problem using the \CLOSEST policy with QoS and bandwidth 
constraints.
For this purpose, we modify an algorithm of Lin, Liu and Wu 
\cite{PangfengLiu06}. Their algorithm {\it Place-replica} is 
used on homogeneous conditions with QoS constraints but without 
bandwidth restrictions. To be able to use the algorithm, we have 
to modify the original platform. We transform the tree $T$ in a 
tree $T^{*}$ by adding a new root $r^{+}$ as father of the 
original root $r$ (see Figure~\ref{fig:transform}). $r^{+}$ is 
connected to $r$ via a link $l_{0}$, where $\BW(l_{0}) = 0$. As 
the bandwidth is limited to $0$, no requests can pass above $r$, 
so that this artificial transformation for computation purposes 
can be adapted to any tree-network.

\begin{figure}
  \centering
  \includegraphics[width=0.25\textwidth]{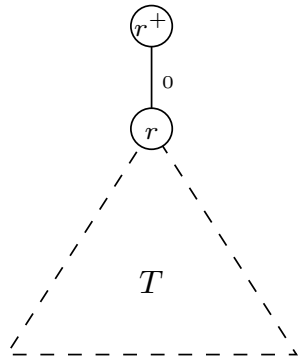}
  \caption{Appearance of $T^{*}$}
  \label{fig:transform}
\end{figure}

 A further, only formal transformation, consists in the 
suppression of clients from the tree and hence the consideration 
of their parents as leaves in the following way 
(Figures~\ref{fig:suppress} and~\ref{fig:kept} give an 
illustration): for every parent $v$ who has only 
leaf-children $v_1, .., v_n$ (i.e., all its children are clients),
 we assign the sum of the requests 
of the $v_j$ as its requests $w(v)$, i.e., 
$ w(v) = \sum_{1 \leq j \leq n}w(v_j)$. The associated QoS is 
set to $(\min_{1\leq j \leq n} \qos(v_j))-1$. This transformation
 is possible, as we use the \CLOSEST policy and hence all children
 have to be treated by the same server. From those parents who 
have some leaf-children $v_1, .., v_n$, but also non-leaf 
children $v_{n+1},..,v_{m}$, the clients can not be suppressed 
completely. In this case the leaf-children $v_1, .., v_n$ are 
compressed to one single client $c$ with requests 
$ w(c) = \sum_{1 \leq j \leq n}w(v_j)$ and QoS 
$\qos(c) = \min_{1\leq j \leq n} \qos(v_j)$. Once again this 
compression is possible due to the restriction on the \CLOSEST 
access policy.

\begin{figure}
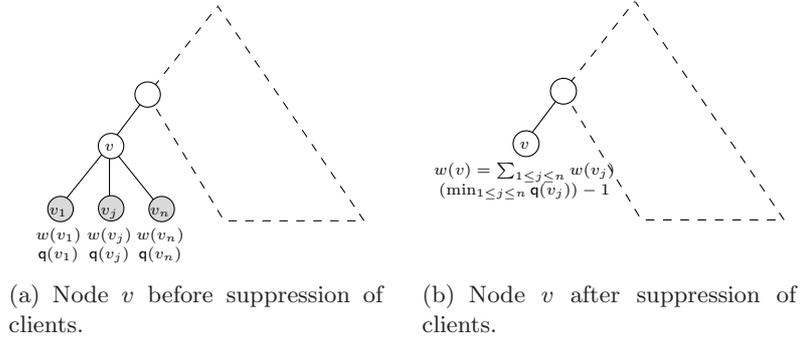

   \centering
   \subfigure[Node $v$ before suppression of clients.]{
 \includegraphics[width=0.3\textwidth]{suppressBefore.fig}
 \label{fig:before}
 }$\quad$
 \subfigure[Node $v$ after suppression of clients.]{
 \includegraphics[width=0.3\textwidth]{suppressAfter.fig}
 \label{fig:after}
 }
 \caption{Suppression of clients}
\label{fig:suppress}
 \end{figure}

\XX works in two phases. In the first phase so called Contribution 
Functions are computed which will serve in the second phase to 
determine the optimal replica placements. In the following some 
new terms are introduced and then the two phases are described in 
detail.

\begin{figure}
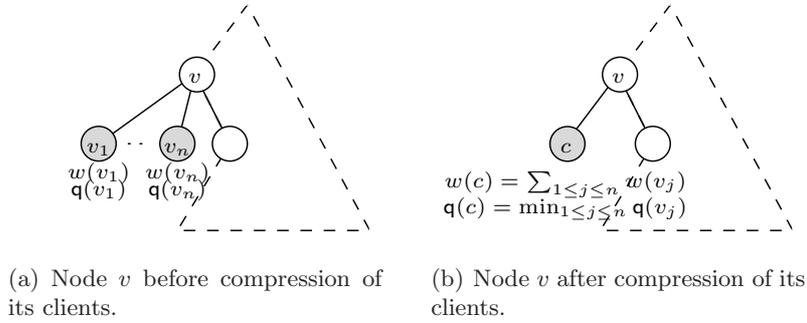

  \centering
  \subfigure[Node $v$ before compression of its clients.]{
    \includegraphics[width=0.3\textwidth]{keptBefore.fig}
  \label{fig:kept_before}
  } $\quad$
  \subfigure[Node $v$ after compression of its clients.]{
    \includegraphics[width=0.3\textwidth]{suppressKept.fig}
    \label{fig:kept_after}
  }
  \caption{Compression of clients.}
  \label{fig:kept}
\end{figure}

\subsubsection{Terminology}
\label{sec:terminology}
Working with a tree $T^{*}$ with root $r^{+}$, we note $t(v)$ the 
subtree rooted by node $v$, and $t'(v) = t(v) - v$, i.e. the 
forest of trees rooted at $v$'s children. The $i$'th ancestor of 
node $v$, traversing the tree up to the root, is denoted by 
$a(v,i)$.

Using these notations, we denote $m(T^{*})$ the minimum 
cardinality set of replicas that has to be placed in tree $T$ such 
that all requests can be treated by a maximum processing capacity 
of $W$ (respecting QoS and bandwidth constraints). In the same 
manner $m(t(v))$ denotes the minimum number of replicas that has 
to be placed in $t'(v)$, such that the remaining requests on node 
$v$ are within $W$. For this purpose we define a contribution 
function $C$. $C(v,i)$ denotes the minimum workload on node 
$a(v,i)$ contributed by $t(v)$  by placing $m(t(v))$ replicas in 
$t'(v)$ and none on $a(v,j)$ for $0 \leq j < i$. The computation 
and an illustrating example are presented below 
(Cf. Section~\ref{sec:ex1} and Section~\ref{sec:ex2}). But before 
we need a last notation. The set $e(v,i)$ denotes the children of 
node $v$ that have to be equipped with a replica such that the 
remaining requests on node $a(v,i)$ are within $W$, there are 
exactly $m(t(v))$ replicas in $t'(v)$ and none on $a(v,j)$ for 
$0 \leq j < i$ and the contribution $t(v)$ on $a(v,i)$ is 
minimized. The computation formula is also given below. Of course 
the compression of leaves is not possible if a client $v$ of the 
original tree $T$ is connected to its father via a communication 
link $l$ that has a lower bandwidth than $v$ requests 
($\BW(l) < w(v)$). In this case we know a priori that there is no 
solution to our problem as $v$'s requests can not be treated.

\begin{figure}
  \centering
  \includegraphics[width=0.3\textwidth]{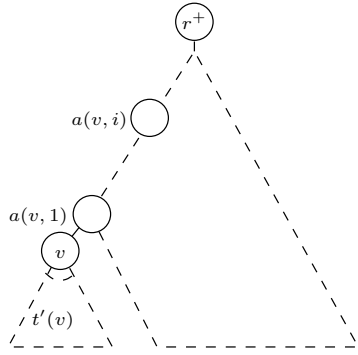}
  \caption{Clarification of the terminology.}
  \label{fig:terminology}
\end{figure}

\subsection{Phase 1: Bottom up computation of set~$e$, amount~$m$ and contribution function~$C$}
\label{sec:phase1}
The computation of $e$, $m$ and $C$ is a bottom up process, 
distinguishing two cases.
\begin{enumerate}
\item {\bf $v$ is a leaf:}

  In this case we do not need $e$ and $m$ and we can directly 
compute the contribution function.
$C(v,i)$ is $w(v)$ when 
$(i \leq \qos(v) \wedge w(v) \leq \min_{\BW}\pat{v}{a(v,i)})$, and 
infinity otherwise.

  We point out that there is no solution if any of the leaves has 
more requests than $W$ or if the bandwidth of any of the clients 
to its parent is not sufficiently high.
  \item {\bf $v$ is an internal node with children $v_{1}, \dots , v_{n}$:}
    \begin{description}
    \item[$i = 0$:]
      If the contribution on $v$ of its children, i.e. the 
incoming requests on $v$ is bigger than the processing capacity of 
inner nodes $W$, we know we have to place some replicas on the 
children to bound the incoming requests on $W$. To find out which 
children have to be equipped with a replica, we take a look at 
the $C(v_{j},1)$-values of the children. The set $e(v,0)$ is used 
to store the $v_{j}$'s that are determined to be equipped with a 
replica. Hence the procedure is the following:
      \begin{itemize}
      \item
	$e(v,0) = \emptyset$
      \item
	while($\sum_{v_{j} \notin e(v,0)} C(v_{j},1) > W$)\\
	add $v_j \in \NN$ with biggest $C(v_{j},1)$ to $e(v,0)$
      \end{itemize}
      Note that the set $\NN$ used in the procedure still 
corresponds to the set of internal nodes of the original tree $T$. 
So we can add leaf nodes of $T^{*}$ that are inner nodes in $T$, 
but we can not add compressed client nodes.
      Note furthermore that there is no client that is added to 
$e(v,0)$. Besides we remark that there is no valid solution within 
$W$ and the present QoS and bandwidth constraints, when all 
children $v_j \in \NN$ of $v$ are equipped with a replica and the 
incoming requests do not fit in $W$. Of course this holds also 
true in the case $i > 0$.
      Subsequently, the value of $m(t(v))$ is determined easily: 
$m(t(v)) = \sum_{1\leq j \leq n}m(t(v_{j})) + |e(v,0)|$. We remind 
that $m(t(v))$ indicates the minimum number of replicas that have 
to be placed in $t'(v)$ to keep the number of contributed requests 
inferior to $W$.
      Finally, the computation of the contribution function :
      $$C(v,0) = \sum_{v_{j} \notin e(v,0)} C(v_{j},1)$$.

    \item[$i > 0$:]
      Treating node $v$, we want to compute the contribution on 
$a(v,i)$. As for $i=0$, we start computing the set $e(v,i)$:
      \begin{itemize}
      \item
	$e(v,i) = \emptyset$
      \item
	while($\sum_{v_{j} \notin e(v,i)} C(v_{j},i+1) > W$)\\
	add $v_j \in \NN$ with biggest $C(v_{j},i+1)$ to $e(v,i)$
      \end{itemize}
      The computation of the contribution function follows a 
similar principle:
      \begin{equation}\label{eq:cont}
	C(v,i) = \begin{cases}\sum_{v_{j} \notin e(v,i)} C(v_{j},i+1) , & \text{if}\; |e(v,i)| = |e(v,0)|\\
	  \infty, & \text{otherwise}\end{cases}
      \end{equation}
      $C(v,i)$ is set to $\infty$, when the number of $|e(v,0)|$ 
replicas placed among the children of $v$ is not sufficient to 
keep the contributed requests on $a(v,i)$ within $W$.

    \end{description}
\end{enumerate}

  \subsection{Example of Phase 1}
\label{sec:ex1}
Consider the tree in Figure~\ref{fig:proofBW} and a processing 
capacity of inner nodes fixed to $W = 15$. The tree has already 
been transformed. So nodes $x$ and $y$ are compressed client-leaves 
(grey scaled in the figure), whereas all other leaves correspond to 
servers (former inner nodes, hence nodes that are within $\NN$). 
We start with the computation of all $C(v,i)$-values of all leaves. 
Leaf $l$ for example has $C(l,0) = 3$ as it holds 3 requests. As 
the link from $l$ to $e$ has a bandwidth of $4$, and the QoS is 
$2$, the requests of $l$ can ascent to node $e$ and hence the 
contribution of $l$'s requests on node $e$, $C(l,1)$, is $3$. In 
the same manner, $C(l,2)$, i.e. the contribution of $l$'s requests 
on node $b$ is $3$ as well. But then the QoS range is exceeded and 
hence the requests of $l$ can not be treated higher in the tree. 
Consequently the contributions on nodes $a$ and $a^{+}$ ($C(l,3)$ 
and $C(l,4)$) are set to infinity. Another example: Leaf $i$ owns 
$7$ requests, but the link from $i$ to its parent $c$ has a lower 
bandwidth, and hence the contribution of $i$ on $c$, $C(i,1)$, has 
to be set to infinity. The whole computation table for the 
leaf-contributions is given in Table~\ref{tab:ex-leaf}.
\begin{figure}
\centering
\includegraphics[width=0.45\textwidth]{ex-proof-bw.fig}
\caption{Example}
\label{fig:proofBW}
\end{figure}

\begin{table}
 \centering
 \begin{tabular}{|r|c c c c c c c c c c c|}
   \hline
   & l & f & x & m & n & h & i & o & p & k & y\\
   \hline
   \hline
   $C(v,0)$ & 3 & 4 & 3 & 2 & 5 & 8 & 7 & 4 & 12 & 3 & 8\\
   $C(v,1)$ & 3 & 4 &  3 & 2 & 5 & 8 & $\infty$ & 4 & 12 & $\infty$ & 8\\
   $C(v,2)$ & 3 & $\infty$ & $\infty$ & 2 & $\infty$ & 8 & $\infty$ & 4 & 12 & $\infty$ & 8 \\
   $C(v,3)$ & $\infty$ & $\infty$ & $\infty$ & 2 & $\infty$ & $\infty$ & $\infty$ & 4 & 12 & $\infty$ & $\infty$\\
   $C(v,4)$ & $\infty$ & & & $\infty$ &  $\infty$ & &  & $\infty$ & $\infty$ & &  \\
   \hline
 \end{tabular}
\caption{Computation of $C(v,i)$-values of leaves.}
\label{tab:ex-leaf}
\end{table}

Table~\ref{tab:ex} is used for the computation of $e$, $m$ and $C$ 
values of inner nodes. During the computation process it is filled 
by main columns, where one main column consists of all inner nodes 
of the same level in the tree. So we start with node $e$. The 
contribution of its child $l$, $C(l,1)$, is $3$ as we computed in 
Table~\ref{tab:ex-leaf}. And as it is the only child, we have that 
the contributed requests on $e$ are less than the processing 
capacity $W$ which is fixed to $15$ and hence we do not need to 
place a replica on the child $l$ of $e$ to minimize the 
contribution on $e$. Corresponding we get $m(t(e)) = 0$, i.e. we 
do not need to place a replica in the subtree $t'(e)$, and a 
contribution $C(e,0) = 3$. $e(e,1)$ and $C(e,1)$ are computed in 
the same manner, taking into account $C(l,2)$. Computing $e(e,2)$, 
i.e. the nodes that have to be equipped with a replica if we want 
to minimize the contribution on node $a(e,2) = a$ by placing 
replicas on the children of $e$ but none on $e$ up to $a$. For 
this purpose we use $C(l,3)$, the contribution of $l$ on $a$ and 
remark that it is infinity. Hence we have $\infty > W = 15$ and 
so we have to equip $l$ with a replica, and as now the set 
$e(e,2)$ has a higher cardinality than $e(e,0)$, we know that this 
solution is not optimal anymore and we set the contribution of 
$C(e,2)$ to infinity (Eq.~\ref{eq:cont}).
Taking a look at node $j$: In the computation of $e(j,0)$, we have 
a total contribution of its children of $16$, which exceeds the 
processing power of $W=15$ (bandwidth and QoS are not restricting 
here). So we know that we have to equip one of the children with a 
replica, and we choose the one with the highest contribution on 
$j$: node $p$. Consequently, we get $m(t(j)) = 1$ as we have to 
place one replica on the children. The contribution $C(j,0)$ 
consists in the $4$ remaining contributed requests of node $o$. 
Once we have finished all computations for this level, we start 
with the computations of the next level, which can be found in the 
next main column of the table. Let us treat exemplarily node $c$. 
The sum of its children's contributions is 
$C(x,1) + C(g,1) + C(h,1) + C(i,1) = \infty$ as 
$C(g,1) = C(i,1) = \infty$. So we add $g$ and $i$ to $e(c,0)$ to 
lower the contributions to $C(x,1) + C(h,1) = 11$ which fits in 
the processing power of $W=15$. The outcome of this is 
$m(t(c)) = 2$ and the remaining requests lead to $C(c,0) = 11$.

\begin{table}
 \centering
 \begin{tabular}{|r|c c c|c c c|c c|}
   \hline
   & e & g & j & b & c & d & a & a$^+$\\
   \hline
   \hline
   $e(v,0)$ & $\emptyset$ & $\emptyset$ & $\{p\}$ & $\emptyset$ & $\{g,i\}$ & $\{k\}$ & $\{b,c\}$ & $\{a\}$\\
   $m(t(v))$ & 0 & 0 & 1 & 0 & 2 & 2 & 6 & 7\\
   $C(v,0)$ & 3 & 7 & 4 & 9 & 11 & 12 & 12 & $\infty$\\
   \hline
   $e(v,1)$ & $\emptyset$ & $\{n\}$ & $\{p\}$ & $\{e\}$ & $\{g,i\}$ & $\{k\}$ & $\{b,c,d\}$ & \\
   $C(v,1)$ & 3 & $\infty$ & 4 & $\infty$ & $\infty$ & 12 & $\infty$ & \\
   \hline
   $e(v,2)$ & $\{l\}$ & $\{n\}$ & $\{p\}$ & $\{e,f\}$ & $\{g,i\}$ & $\{j,k\}$ &  & \\
   $C(v,1)$ & $\infty$ & $\infty$ & 4 & $\infty$ & $\infty$ & $\infty$ & & \\
   \hline
   $e(v,3)$ & $\{l\}$ & $\{m,n\}$ & $\{o,p\}$ &  &  &  &  & \\
   $C(v,1)$ & $\infty$ & $\infty$ & $\infty$ &  & & & & \\
   \hline
 \end{tabular}
 \caption{Computation of $e$, $m$ and $C$ for internal nodes.}
\label{tab:ex}
\end{table}

\subsection{Phase 2: Top down replica placement}
\label{sec:phase2}
The second phase uses the precomputed results of the first phase 
to decide about the nodes on which to place a replica.
The goal is to place $m(T^{*}) = m(t(r^{+}))$ replicas in 
$t'(r^{+})$. Note that this means that there is no replica on 
$r^{+}$ and hence only the original tree $T$ will be equipped with 
replicas. If the workload on node $r$ is within $W$, we have a 
feasible solution.

Phase~2 is a recursive approach. Starting with $i = 0$ on node 
$v = r^{+}$, all nodes that are within
$e(v,i)$ are equipped with a replica. In this top down approach, 
$i$ indicates the distance of node $v$ to its first ancestor up 
in the tree that is equipped with a replica and hence the set 
$e(v,i)$ denotes the set of children of $v$ that have to be 
equipped with a replica in order to minimize the contribution of 
$v$ on $a(v,i)$.
Next the procedure is called recursively with the appropriate 
index $i$.
Algorithm~\ref{algo:top-down} gives the pseudo-code for the top 
down placement phase, which is the same as the one in~\cite{PangfengLiu06}.

\begin{algorithm}
\label{algo:top-down}
\caption{Top down replica placement}
procedure {\bf Place-replica} (v, i)\\
\Begin{
    \If{$v \in \CC$}
       {
	 return\;
       }
       place a replica at each node of $e(v,i)$\;
       \ForAll{$c \in \child(v)$}
	      {
		\eIf{$c \in e(v,i)$}
		    {
		      Place-replica(c,0)\;
		    }
		    {
		      Place-replica(c,i+1)\;
		    }
	      }
  }
\end{algorithm}

\subsection{Example of Phase 2}
\label{sec:ex2}
We start with the results of Phase 1 (Cf. Table~\ref{tab:ex-leaf} 
and~\ref{tab:ex}) and call then the procedure Place-replica 
(Algorithm~\ref{algo:top-down}) with $(a^{+},0)$. $a^{+}$ is not 
a leaf, so we place a replica on its child $a$, as $a \in e(a,0)$ 
and then recall the procedure with $(a,0)$. This time we place 
replicas on $b$ and $c$ and call the procedure with values 
$(b,0), (c,0)$ and $(d,1)$. We have to increment $i$ to $1$ when 
we treat node $d$, as we already know that we will not equip $d$ 
with a replica, and hence the children of $d$ might give their 
contribution directly to $a$. So we have to examine which of the 
children of $d$ have to be equipped with a replica, to minimize 
the contribution on $a$. This is stored in the $e(v_j,1)$-values 
of all children $v_j$ of $d$. So every time we do not place a 
replica on a node and descent to its children, we increase the 
distance-indicator $i$ to the first replica that can be found the 
way up to the root. The recursive procedure call for the entire 
example is given in Table~\ref{tab:rec-call}. PR(x,i) stands for 
the call of Place-replica with parameters (x,i) and $\rightarrow x$ 
indicates that node $x$ is equipped with a replica.

\begin{table}
  \begin{tabular}{l l l l l l l l l}
    PR(a$^{+}$,0) & & & & & & & \\
    \hline
    $\rightarrow$ {\bf a} & & & & & & & \\
    PR(a,0) & & & & & & & \\
    \hline
    $\rightarrow$  {\bf b,c} & & & & & & & \\
    PR(b,0) & & PR(c,0) & & & & PR(d,1) & & \\
    \hline
    & & $\rightarrow$ {\bf g,i} & & & & $\rightarrow$ {\bf k} & & \\
    PR(e,1) & PR(f,1) & PR(g,0) & & PR(h,2) & PR(i,0) & PR(j,2) & & PR(k,0) \\
    \hline
    & ret &  & & ret & ret &  $\rightarrow$ {\bf p} & & ret \\
    PR(l,2) & & PR(m,1) & PR(n,0)& & & PR(o,3) & PR(p,0) &  \\
    \hline
    ret & & ret & ret & & & ret & ret & \\
    \hline
  \end{tabular}
\label{tab:rec-call}
\caption{Scheme on the recursive calls of the procedure \textbf{Place-replica}}
\end{table}

\subsection{Complexity}
\label{sec:complexity}
Let us take a look on the complexity of \XX. For each node $v$ we 
have to compute $e$, $m$ and $C$ values. So the computation 
requires $n \log n$, if $v$ has $n$ children and if we sort the 
$C$ values from all of $v$'s children. We have to do at most L 
sorting, where L is the maximum range limit among all nodes. So 
at all the computation complexity for the values for one node is 
$L n \log n$, and we get a total complexity of $L N \log N$, 
where $N$ is the number of nodes in the tree.

\subsection{Optimality}
\label{sec:optimality}
In this section we prove optimality of our algorithm \XX by 
recursion over levels. For this purpose we apply a theorem 
introduced by Liu et al.~\cite{PangfengLiu06} and presented below 
as Theorem~\ref{th:opt}. Liu et al. used this theorem in order to 
prove the existence of an optimal solution on a homogeneous data 
grid tree under QoS constraints. As the theorem does not take into 
account if there are any constraints like QoS or bandwidth, we can 
adopt it for our problem.

\begin{theorem}
\label{th:opt}
Consider a data grid tree $T$, a node $v$ in $T$ with children 
$v_1,..,v_n$ and a workload $W$. There exists a replica set $R$ so 
that $|R| = m(T)$, $R$ minimizes the total workload due to $R$ 
from $t'(v)$ on $a(v,i)$ for $i \geq 1$, and 
$|R \cap t'(v_j)| = m(t(v_j))$.
\end{theorem}

In other words, Theorem~\ref{th:opt} guarantees that for a tree 
$T$ with fixed processing capacity $W$ there exists a replica set 
$R$ whose cardinality is the minimum number of replicas that has 
to be placed in $t'(r)$ (where $r$ is the root of $T$), such that 
the remaining requests on $r$ are within $W$. Furthermore for a 
node $v$ with children $v_1,...,v_n$, due to $R$ the workload on a 
ancestor $a(v,i)$ of $v$ is minimized and the number of replicas 
that are placed in the subtree $t'(v_j)$ is minimal.

\begin{proof}
We can use the same arguments as Liu et al. as we did not change 
the definition of $m$-values but the constraints on $m$.
By definition of $m(t(v_j))$, we know that this is the minimal 
number of replicas that has to be placed in $t'(v_j)$ such that 
the contribution on $v_j$ is within $W$. Hence $|R \cap t'(v_j)|$ 
can not be less than $m(t(v_j))$ because otherwise the 
contribution on $v_j$ would exceed $W$. On the other side in any 
optimal solution for $t(v_j)$, we can not place more replicas in 
$t'(v_j)$ than $m(t(v_j))$ and than one more on $v_j$. The 
resulting contribution on $a(v,i)$ decreases at most when placing 
the replica on $v_j$.
\end{proof}

\begin{theorem}
Algorithm \XX returns an optimal solution to the \REP problem with 
fixed $W$, QoS and bandwidth constraints, if there exists a 
solution.
\end{theorem}

\begin{proof}
We perform an induction over levels to prove optimality. We 
consider any tree $T^{*}$ of hight $n+1$ and start at level~0, 
which consists in the artificial root $r^{+}$ 
(Cf. Figure~\ref{fig:inductionTopDown}).

\begin{figure}
\centering
\includegraphics[width=0.55\textwidth]{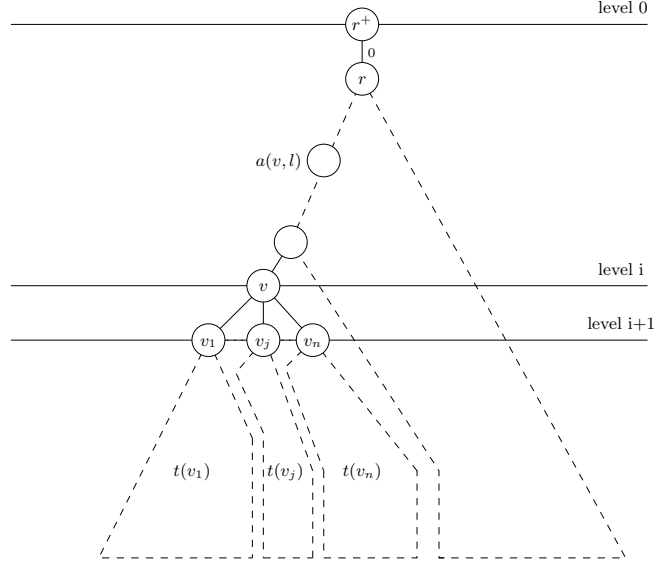}
\caption{Induction over levels.}
\label{fig:inductionTopDown}
\end{figure}

\begin{description}
\item[level 0:] Using Theorem~\ref{th:opt}, we know that there 
exists an optimal solution $R_0$ for our tree (i.e., a set $R$ of 
replicas whose cardinality is $m(T^{*})$) such that 
$|R_{0} \cap t'(r)| = m(t(r))$. We have 
$m(T^{*}) = m(t(r)) + |e(r^{+},0)|$ by definition of $e(r^{+},0)$. 
Hence $e(r^{+},0) = \{r\}$ if and only if $r \in R_{0}$. This is 
exactly how the algorithm pursuits.

\item[level i $\rightarrow$ i+1:]
We assume that we have placed the replicas from level~0 to 
level~i (with Algorithm~\ref{algo:top-down}) and that there exists 
an optimal solution $R_i$ with these replicas. We further suppose 
that for each node $v$ in level $i$ it holds 
$|R_{i} \cap t'(v)| = m(t(v))$. Let us consider a node $v$ in 
level $i$ with children $v_1, .., v_n$ and we define 
$l:= \min\{k\geq 0 | a(v,k) \in R_i\}$. In the next step of the 
algorithm we equip the elements of $e(v,l)$ with a replica. We 
have $m(t(v)) = \sum_{1\leq j \leq n}m(t(v_j))+|e(v,0)|$, i.e. the 
minimal number of replicas in the subtrees $t'(v_j)$ and the 
minimal number of replicas on the children of $v$ that have to be 
placed to keep the contributed requests on $v$ within $W$. By 
definition of $e(v,l)$ we have that 
$|\{j \in \{1, .., n\} | v_j \in R_i\}| \geq |e(v,l)|$ and we also
 have $|e(v,l)| = |e(v,0)|$ as the contribution $C(v,l)$ is finite 
and $R_i$ a solution. For the inequality, there is even equality 
because otherwise there would exist a $j$ such that 
$|t'(v_j) \cap R_i| < m(t(v_j))$, which is impossible. With this 
equality, we can replace the children of $v$ that are in $R_i$ by 
the children of $v$ that are in $e(v,l)$ creating a solution 
$R_{i+1}$. So $R_{i+1}$ is also an optimal solution, because 
$|R_i| = |R_{i+1}|$ (we did not change the nodes of the other 
levels) and the contribution of $t(v)$ on $a(v,l)$ has at most 
decreased. Furthermore for every node $v'$ at level i+1 we have 
$|R_{i+1} \cap t'(v')| = m(t(v'))$.
\end{description}

So the last solution $R_{n}$ that we get in the induction step $n$ 
is optimal and it corresponds to the solution that we obtain by 
our algorithm.
\end{proof}

\section{Conclusion}
\label{sec:conclusion}
In this paper we dealt with the \REP optimization problem with QoS 
and bandwidth constraints. We restricted our research on 
\CLOSEST/Homogeneous instances. We were able to prove 
polynomiality and proposed the optimal algorithm \XX. This 
algorithm extends an existing algorithm in two important areas. 
First the set of clients and the set of servers can be distinct 
now and does not require exclusively double-functionality nodes 
anymore. The other major contribution is the expansion to the 
interplay of different nature constraints. QoS, which is a proper 
constraint for each client, and bandwidth, a global resource 
limitation, subordinate to a common optimization function. This 
accomplishment completes furthermore the study on complexity of 
\CLOSEST/Homogeneous in tree networks.
\bibliographystyle{abbrv}
\bibliography{biblio}

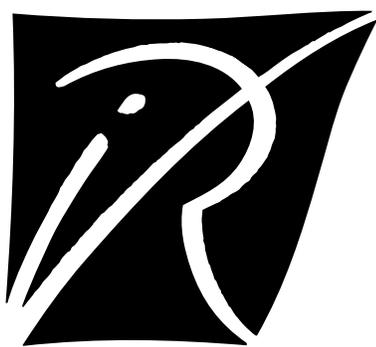
\end{document}